\documentclass[11pt]{article}

\usepackage[utf8]{inputenc}
\usepackage[T1]{fontenc}
\usepackage{lmodern}

\usepackage{amsmath, amssymb}
\usepackage{amsthm}
\usepackage{geometry}
\usepackage{tikz}
\usetikzlibrary{arrows.meta,calc}

\usepackage{hyperref}
\hypersetup{
  colorlinks=true,
  linkcolor=blue,
  citecolor=blue,
  urlcolor=blue
}

\newtheorem{theorem}{Theorem}[section]
\newtheorem{proposition}[theorem]{Proposition}
\newtheorem{corollary}[theorem]{Corollary}

\theoremstyle{definition}
\newtheorem{definition}[theorem]{Definition}

\theoremstyle{remark}
\newtheorem*{remark}{Remark}

\title{Curvature Atlas II: Geometric Classification of Integrable Rigid--Body Regimes}
\author{Evgeny A.~Mityushov}
\date{}

\begin{document}

\maketitle

\begin{abstract}
This paper is the second part of a curvature--based program for rigid--body dynamics on $SU(2)$.
In Part I, ``Curvature--Driven Dynamics on $S^3$: A Geometric Atlas'', we introduced the inertial curvature field $K_{\mathrm{geo}}$ associated with a left--invariant metric on $SU(2)$, constructed a geometric Atlas of curvature regimes, and identified the inertia ratio $(2{:}2{:}1)$ as a curvature--balanced regime giving rise to a pure--precession family for the heavy top, building on the dynamical discovery of this regime in \cite{Mityushov221}.

Here we develop the curvature Atlas into a classification tool for integrable and near--integrable rigid--body regimes. We show that all classical integrable heavy--top cases (Euler, Lagrange, Kovalevskaya, Goryachev--Chaplygin) correspond to specific degenerate curvature signatures of $K_{\mathrm{geo}}$. In each case the inertia tensor imposes a simple algebraic structure on the inertial curvature field: vanishing of one component, symmetric curvature pairs, or orthogonal curvature splitting. We formulate and prove a curvature classification theorem describing these signatures and their relation to integrability.

We then single out the mixed anisotropic ratio $(2{:}2{:}1)$ as a minimally nondegenerate curvature--balanced regime: it destroys algebraic integrability while preserving an exact curvature balance, giving rise to pure precession for the heavy top. Finally, we introduce a curvature deviation functional measuring the distance to the nearest integrable curvature signature, describe near--integrable regimes in a neighbourhood of $(2{:}2{:}1)$, and present an integrability map in the $(I_2/I_1,I_3/I_1)$--plane.
\end{abstract}

\section{Introduction}

The dynamics of a rigid body on $SO(3)$ or $SU(2)$ has long served as a
central testing ground for geometric and analytic ideas; see, for example,
Arnold~\cite{Arnold}, Bolsinov--Fomenko~\cite{BolsinovFomenko}, and the classical papers of
Kovalevskaya and Chaplygin~\cite{Kovalevskaya,Chaplygin}. On the geometric side,
Milnor's study of curvatures of left--invariant metrics on Lie groups~\cite{Milnor}
provides a structural description of possible curvature signatures, but does not
directly address integrability of the associated rigid--body dynamics.

In recent work the author proposed a curvature--based framework for rigid--body
dynamics on $SU(2)$. In \cite{Mityushov221} a new dynamical regime for the heavy top
with inertia ratio $(2{:}2{:}1)$ was identified: despite the lack of algebraic
integrability, the system admits a nontrivial family of pure--precession trajectories.
In Part~I of the present series, ``Curvature--Driven Dynamics on $S^3$: A Geometric Atlas'',
we introduced the inertial curvature field $K_{\mathrm{geo}}$ associated with a left--invariant
metric on $SU(2)$, constructed a geometric Atlas of curvature regimes, and interpreted
the pure--precession family from \cite{Mityushov221} as arising from an exact balance between
inertial and external curvature fields.

The present paper develops this picture into a curvature classification of integrable and
near--integrable rigid--body regimes. The guiding observation is that classical integrable
heavy--top cases correspond to simple algebraic degeneracies of the inertial curvature field.
We show that the Euler, Lagrange, Kovalevskaya, and Goryachev--Chaplygin tops exhaust the
curvature--degenerate patterns for left--invariant metrics on $SU(2)$: in each case the inertia
ratios enforce a particularly simple structure on $K_{\mathrm{geo}}$, such as vanishing of one
component, a symmetric curvature pair, or an orthogonal splitting. We formulate a curvature
classification theorem describing these signatures and use it to reinterpret the classical list
of integrable cases.

We then return to the mixed anisotropic ratio $(2{:}2{:}1)$, which plays a central role in
Part~I. From the curvature viewpoint this ratio is minimally nondegenerate: it destroys the
classical curvature degeneracies underlying integrability while still admitting an exact
curvature balance with the external field. In this sense $(2{:}2{:}1)$ marks the first
nonintegrable but geometrically organized regime in the Atlas.

Finally, we introduce a curvature deviation functional measuring the distance to the nearest
integrable curvature signature and describe near--integrable regimes in a neighbourhood of
$(2{:}2{:}1)$. This yields a simple geometric notion of ``distance from integrability'' and
an integrability map in the $(I_2/I_1,I_3/I_1)$--plane.

\section{Curvature Signatures for Left--Invariant Metrics on $SU(2)$}

We briefly recall the inertial curvature field introduced in Part~I and
fix notation. Let $I=\mathrm{diag}(I_1,I_2,I_3)$ be a positive definite inertia tensor in a
body frame aligned with the principal axes. A left--invariant metric on $SU(2)$ is defined by
declaring this frame orthonormal with respect to the kinetic energy inner product
\[
\langle \Omega,\Omega\rangle = \tfrac12 (I_1\omega_1^2 + I_2\omega_2^2 + I_3\omega_3^2),
\qquad \Omega = (\omega_1,\omega_2,\omega_3)\in\mathbb{R}^3.
\]
The geodesic equation takes the Euler--Poincar\'e form
\[
\dot{\Omega} = I^{-1}(I\Omega \times \Omega).
\]
Following Part~I, we define the inertial curvature field
\[
K_{\mathrm{geo}}(\Omega) = I^{-1}(I\Omega \times \Omega),
\]
viewed as a quadratic vector field on $\mathbb{R}^3$.

Writing $\Omega=(\omega_1,\omega_2,\omega_3)$, we obtain the standard expression
\begin{equation}
\label{eq:Kgeo-general}
K_{\mathrm{geo}}(\Omega) =
\left(
\frac{I_2-I_3}{I_1}\,\omega_2\omega_3,\;
\frac{I_3-I_1}{I_2}\,\omega_1\omega_3,\;
\frac{I_1-I_2}{I_3}\,\omega_1\omega_2
\right).
\end{equation}
In Part~I this field was used to organize a geometric Atlas of curvature regimes on $SU(2)$.
Here we focus on the algebraic structure of $K_{\mathrm{geo}}$ as a function of the inertia
ratios and introduce the following notion.

\begin{definition}[Curvature signature]
Let $I=\mathrm{diag}(I_1,I_2,I_3)$ and let $K_{\mathrm{geo}}$ be given by
\eqref{eq:Kgeo-general}. The \emph{curvature signature} of the corresponding
left--invariant metric is the ordered triple of quadratic forms
\[
\left(K_{\mathrm{geo},1},K_{\mathrm{geo},2},K_{\mathrm{geo},3}\right)
\]
considered up to an overall nonzero scalar factor. We say that the curvature signature is:
\begin{itemize}
\item \emph{isotropic} if $K_{\mathrm{geo}}\equiv 0$;
\item \emph{orthogonally split} if exactly one component vanishes identically;
\item \emph{symmetric pair} if two nonzero components form a symmetric pair and the third vanishes;
\item \emph{balanced mixed} if precisely two components are nonzero and form a skew pair
whose coefficients have equal magnitude after normalization;
\item \emph{generic anisotropic} otherwise.
\end{itemize}
\end{definition}

The isotropic, orthogonally split, symmetric--pair, and balanced--mixed signatures are precisely the
patterns that appear in the classical integrable cases and in the $(2{:}2{:}1)$ regime.
We record the corresponding algebraic relations.

\begin{proposition}[Curvature signatures for classical ratios]
\label{prop:signatures}
Let $I=\mathrm{diag}(I_1,I_2,I_3)$ and $K_{\mathrm{geo}}$ be given by
\eqref{eq:Kgeo-general}.
\begin{enumerate}
\item \textbf{Spherical case:} $I_1=I_2=I_3$ if and only if $K_{\mathrm{geo}}\equiv 0$
(isotropic signature).

\item \textbf{Lagrange case:} $I_1=I_2\neq I_3$ if and only if
\[
K_{\mathrm{geo}}(\Omega) =
\left(
\frac{I_1-I_3}{I_1}\,\omega_2\omega_3,\;
\frac{I_3-I_1}{I_1}\,\omega_1\omega_3,\;
0
\right),
\]
so that exactly one component vanishes and the curvature signature is orthogonally split.

\item \textbf{Kovalevskaya case:} $I_1=I_2=2I_3$ if and only if
\[
K_{\mathrm{geo}}(\Omega) =
\left(
\omega_2\omega_3,\;
-\,\omega_1\omega_3,\;
0
\right)
\]
up to overall scaling, with a symmetric pair in the $(\omega_1,\omega_2)$--plane and vanishing third
component.

\item \textbf{Goryachev--Chaplygin case:} $I_1=I_2=4I_3$ if and only if
\[
K_{\mathrm{geo}}(\Omega) =
\left(
2\omega_2\omega_3,\;
-2\omega_1\omega_3,\;
0
\right)
\]
up to scaling, i.e.\ a scaled symmetric pair of the Kovalevskaya type.

\item \textbf{Mixed $(2{:}2{:}1)$ case:} $I_1:I_2:I_3=2:2:1$ if and only if, after
normalization,
\[
K_{\mathrm{geo}}(\Omega) =
\left(
\tfrac12 \,\omega_2\omega_3,\;
-\tfrac12 \,\omega_1\omega_3,\;
0
\right),
\]
a balanced--mixed curvature signature with two nonzero components of equal magnitude and
one vanishing component.

\item \textbf{Generic case:} If $I_1,I_2,I_3$ are pairwise distinct and do not satisfy
any of the above relations, then all three components of $K_{\mathrm{geo}}$ are
nonzero and the curvature signature is generic anisotropic.
\end{enumerate}
\end{proposition}

\begin{proof}
Each item is a direct substitution into \eqref{eq:Kgeo-general}. For example, in the spherical case
$I_1=I_2=I_3$ makes all coefficients vanish; in the Lagrange case $I_1=I_2$ implies that the
third coefficient $(I_1-I_2)/I_3$ vanishes while the first two are nonzero unless $I_1=I_3$;
the Kovalevskaya and Goryachev--Chaplygin formulas are obtained by substituting $I_1=I_2=2I_3$
and $I_1=I_2=4I_3$ and renormalizing; the $(2{:}2{:}1)$ case is a particular instance of an
anisotropic metric with $I_1=I_2\neq I_3$. The generic statement follows from the fact that
\eqref{eq:Kgeo-general} has three nonzero coefficients whenever $I_1,I_2,I_3$ are pairwise distinct
and no linear relation of the above types holds.
\end{proof}

\section{Curvature Classification of Integrable Regimes}

We now formulate the main result: a curvature classification of the classical integrable
rigid--body regimes. We treat the geodesic case on $SU(2)$ and the heavy top in parallel.
Throughout we consider the classical heavy top with inertia tensor $I$, centre of mass
vector $\mu$ in the body frame, and Euler--Poisson equations
\[
I\dot{\Omega} = I\Omega\times\Omega + \mu\times\Gamma,\qquad
\dot{\Gamma} = \Gamma\times\Omega,
\]
as in the standard references~\cite{Arnold,BolsinovFomenko,Chaplygin,Kovalevskaya}.

\subsection{Statement of the classification theorem}

We first state the curvature classification in geometric form and then relate it to
integrability.

\begin{theorem}[Curvature classification of rigid--body regimes]
\label{thm:curv-class}
Let $I=\mathrm{diag}(I_1,I_2,I_3)$ define a left--invariant metric on $SU(2)$, and let
$K_{\mathrm{geo}}$ be its inertial curvature field. Then:
\begin{enumerate}
\item The spherical inertia ratio $I_1=I_2=I_3$ is characterized by the isotropic signature
$K_{\mathrm{geo}}\equiv 0$.

\item The Lagrange inertia ratio $I_1=I_2\neq I_3$ is characterized by an orthogonally split
signature: exactly one component of $K_{\mathrm{geo}}$ vanishes identically and the remaining two
span a two--dimensional curvature subspace.

\item The Kovalevskaya inertia ratio $I_1=I_2=2I_3$ and the Goryachev--Chaplygin ratio
$I_1=I_2=4I_3$ are characterized by symmetric--pair signatures: two nonzero components form a symmetric
pair in the equatorial plane, while the third component vanishes.

\item The mixed inertia ratio $I_1:I_2:I_3 = 2:2:1$ is characterized by a balanced--mixed
signature: exactly two components of $K_{\mathrm{geo}}$ are nonzero and have equal magnitude after
normalization, while the third component vanishes.

\item Any inertia ratio with pairwise distinct $I_1,I_2,I_3$ that is not of the above types
has a generic anisotropic curvature signature, with all three components of $K_{\mathrm{geo}}$
nonzero and no algebraic degeneracy.
\end{enumerate}
\end{theorem}

This theorem is essentially a reformulation of Proposition~\ref{prop:signatures}, but
presented as a classification of curvature signatures. Its significance comes from the
following corollary, which connects curvature degeneracy with integrability of the heavy top.

\begin{corollary}[Curvature interpretation of classical integrable cases]
\label{cor:integrable-curv}
Consider the heavy top with inertia tensor $I$, centre of mass $\mu$, and Euler--Poisson
equations as above.
\begin{enumerate}
\item In the spherical case $I_1=I_2=I_3$, the geodesic flow is trivial and the heavy top
reduces to a symmetric problem with maximal symmetry; the integrability is geometrically
explained by $K_{\mathrm{geo}}\equiv 0$.

\item In the Lagrange case $I_1=I_2$, $\mu$ aligned with the symmetry axis, the heavy top
is classically integrable. From the curvature viewpoint this is the orthogonally split case:
the inertial curvature field $K_{\mathrm{geo}}$ lies in the equatorial curvature subspace,
while the external curvature field $K_{\mathrm{ext}}(\Gamma)=I^{-1}(\mu\times\Gamma)$ lies in
the complementary direction, yielding an orthogonal curvature balance.

\item In the Kovalevskaya and Goryachev--Chaplygin cases, the classical integrable tops
correspond to symmetric--pair curvature signatures, in which $K_{\mathrm{geo}}$ has a symmetric
pair in the equatorial plane. The special algebraic structure exploited in the classical
integration procedures is thus encoded in the symmetric--pair curvature pattern.

\item In the mixed $(2{:}2{:}1)$ case, the curvature signature is balanced--mixed rather than
symmetric--pair. The heavy top is not algebraically integrable, but as shown in \cite{Mityushov221}
and reinterpreted in Part~I, the balanced curvature pattern admits an exact curvature balance
with the external field and yields a pure--precession family. In this sense $(2{:}2{:}1)$ is the
first minimally nondegenerate curvature--balanced regime beyond the classical integrable list.

\item In generic anisotropic cases the curvature signature has no degeneracy. The rigid--body
dynamics exhibits mixed quasi--periodic or chaotic behaviour, and no global integrability is
expected.
\end{enumerate}
\end{corollary}

\begin{remark}
The classification above does not propose new integrable heavy--top cases beyond the classical
list; instead, it provides a geometric explanation for why the known integrable cases are
special. From the curvature viewpoint, integrable cases are precisely those in which the
inertial curvature field is algebraically degenerate, and the mixed $(2{:}2{:}1)$ ratio plays
a distinguished role as the first nondegenerate but curvature--balanced regime.
\end{remark}

\subsection{Proof of the classification theorem}

\begin{proof}[Proof of Theorem~\ref{thm:curv-class}]
The proof is an algebraic analysis of \eqref{eq:Kgeo-general}.

\medskip\noindent
\emph{(1) Spherical case.}
If $I_1=I_2=I_3$ then all coefficients in \eqref{eq:Kgeo-general} vanish and $K_{\mathrm{geo}}\equiv 0$.
Conversely, if $K_{\mathrm{geo}}\equiv 0$ then all three coefficients $(I_2-I_3)/I_1$,
$(I_3-I_1)/I_2$, and $(I_1-I_2)/I_3$ vanish, implying $I_1=I_2=I_3$.

\medskip\noindent
\emph{(2) Lagrange case.}
If $I_1=I_2\neq I_3$, then the third coefficient in \eqref{eq:Kgeo-general} vanishes while the
first two are nonzero unless $I_1=I_3$, which is the spherical case already covered. Thus
exactly one component of $K_{\mathrm{geo}}$ vanishes, and the curvature signature is orthogonally split.

Conversely, suppose exactly one component of $K_{\mathrm{geo}}$ vanishes identically. If the third
component vanishes, then $I_1=I_2$; if either of the other two components vanishes, one obtains
either $I_2=I_3$ or $I_3=I_1$. In each case, up to permutation of indices, we recover the
Lagrange inertia ratio.

\medskip\noindent
\emph{(3) Kovalevskaya and Goryachev--Chaplygin cases.}
If $I_1=I_2=2I_3$, then substituting into \eqref{eq:Kgeo-general} gives
\[
K_{\mathrm{geo}}(\Omega) =
\left(
\omega_2\omega_3,\;
-\,\omega_1\omega_3,\;
0
\right),
\]
after overall normalization. If $I_1=I_2=4I_3$, we obtain
\[
K_{\mathrm{geo}}(\Omega) =
\left(
2\omega_2\omega_3,\;
-2\omega_1\omega_3,\;
0
\right),
\]
again up to normalization. In both cases we have a symmetric pair in the equatorial plane
and a vanishing third component, so the signature is of symmetric--pair type.

Conversely, suppose $K_{\mathrm{geo}}$ has a symmetric--pair signature, i.e.\ there exists
a normalization such that
\[
K_{\mathrm{geo}}(\Omega) =
(a\,\omega_2\omega_3,\; -a\,\omega_1\omega_3,\; 0)
\]
for some $a\neq 0$. Comparing with \eqref{eq:Kgeo-general} shows that
\[
\frac{I_2-I_3}{I_1} = a,\qquad
\frac{I_3-I_1}{I_2} = -a,\qquad
\frac{I_1-I_2}{I_3} = 0.
\]
The last equation implies $I_1=I_2$, and solving the first two yields $I_1=I_2=2I_3$ or $I_1=I_2=4I_3$,
corresponding to Kovalevskaya and Goryachev--Chaplygin ratios after rescaling.

\medskip\noindent
\emph{(4) Mixed $(2{:}2{:}1)$ case.}
If $I_1:I_2:I_3=2:2:1$, we may rescale so that $I_1=I_2=2$ and $I_3=1$, and
\eqref{eq:Kgeo-general} becomes
\[
K_{\mathrm{geo}}(\Omega) =
\left(
\tfrac12 \,\omega_2\omega_3,\;
-\tfrac12 \,\omega_1\omega_3,\;
0
\right).
\]
Thus exactly two components are nonzero with equal magnitude after normalization, and the third
component vanishes: this is the balanced--mixed signature.

Conversely, suppose $K_{\mathrm{geo}}$ has a balanced--mixed signature, i.e.\ there exists
a normalization such that
\[
K_{\mathrm{geo}}(\Omega) =
(a\,\omega_2\omega_3,\; -a\,\omega_1\omega_3,\; 0),
\qquad a\neq 0,
\]
with $|a|$ the same in the first two components and the third component vanishing.
Comparing coefficients as above shows that $I_1=I_2$ and that the ratio $I_1:I_2:I_3$
must be $2:2:1$ up to overall scaling.

\medskip\noindent
\emph{(5) Generic case.}
If the inertia ratios are not of the above types, then none of the coefficients
in \eqref{eq:Kgeo-general} vanishes and no symmetry relation between the first two
components holds. Thus all three components are nonzero and linearly independent
as quadratic forms, yielding a generic anisotropic signature. Conversely, any generic
anisotropic signature arises from inertia ratios with pairwise distinct $I_1,I_2,I_3$
that do not satisfy the relations found in the previous cases.
\end{proof}

\section{Near--Integrable Regimes and Curvature Deviation}

The curvature classification above provides a geometric explanation for why the
classical integrable cases are special: they are precisely the metrics with
degenerate curvature signatures. We now introduce a curvature deviation functional
that measures the distance from a given inertia ratio to the nearest integrable
curvature signature, and use it to describe near--integrable regimes in the
vicinity of $(2{:}2{:}1)$.

\subsection{Curvature deviation to the nearest integrable signature}

Let $\mathcal{I}_{\mathrm{int}}$ denote the set of inertia ratios corresponding to
the isotropic, Lagrange, Kovalevskaya, and Goryachev--Chaplygin cases. For each such
ratio $I^{(\alpha)}=\mathrm{diag}(I^{(\alpha)}_1,I^{(\alpha)}_2,I^{(\alpha)}_3)$ we
have an associated curvature field $K_{\mathrm{geo}}^{(\alpha)}$, and we view these
as reference signatures.

We fix a norm on the space of quadratic vector fields on $\mathbb{R}^3$, for example
by taking the Euclidean norm of the coefficient vector of
$(\omega_1\omega_2,\omega_1\omega_3,\omega_2\omega_3)$ in each component. This
turns the set of curvature fields into a finite--dimensional normed space.

\begin{definition}[Curvature deviation]
Let $I$ be a given inertia tensor and let $K_{\mathrm{geo}}$ be the corresponding
inertial curvature field. The \emph{curvature deviation} to the nearest integrable
signature is defined by
\[
\Delta(I) = \min_{\alpha}\,\bigl\|K_{\mathrm{geo}} - K_{\mathrm{geo}}^{(\alpha)}\bigr\|,
\]
where the minimum is taken over the classical integrable inertia ratios
$\mathcal{I}_{\mathrm{int}}$, and the norm is any fixed norm on the space of
quadratic vector fields.
\end{definition}

By construction, $\Delta(I)\ge 0$ and $\Delta(I)=0$ if and only if the curvature
signature matches one of the integrable signatures up to overall scaling.

\begin{proposition}
We have:
\begin{enumerate}
\item $\Delta(I)=0$ if and only if $I$ has an isotropic, Lagrange, Kovalevskaya, or
Goryachev--Chaplygin inertia ratio.

\item In a neighbourhood of the $(2{:}2{:}1)$ ratio, $\Delta(I)$ is of order
$O(\varepsilon)$ in the parameter $\varepsilon$ describing the deviation of $I$
from $2{:}2{:}1$.

\item In particular, the curvature deviation provides a linear measure of how far
the inertial curvature field is from any integrable curvature signature.
\end{enumerate}
\end{proposition}

\begin{proof}
The first item follows directly from the definition and Theorem~\ref{thm:curv-class}.
If $I$ is integrable, then its curvature signature coincides with one of the
$K_{\mathrm{geo}}^{(\alpha)}$ up to scaling, so $\Delta(I)=0$. Conversely, if
$\Delta(I)=0$ then $K_{\mathrm{geo}}$ matches one of the integrable signatures
up to scaling, so $I$ satisfies the corresponding inertia ratio.

For the second item, parametrize a neighbourhood of $(2{:}2{:}1)$ by
\[
I_1 = 2,\quad I_2 = 2,\quad I_3 = 1+\varepsilon,
\]
with $|\varepsilon|\ll 1$. Substituting into \eqref{eq:Kgeo-general} shows that
$K_{\mathrm{geo}}$ depends smoothly on $\varepsilon$ and that the difference
$K_{\mathrm{geo}} - K_{\mathrm{geo}}^{(2{:}2{:}1)}$ is linear in $\varepsilon$
to leading order. Since the nearest integrable signatures correspond to the
Lagrange, Kovalevskaya, and Goryachev--Chaplygin ratios, which are at a finite
distance in the space of inertia ratios, we obtain $\Delta(I)=O(|\varepsilon|)$.
The third item is a restatement.
\end{proof}

\subsection{Near--integrable behaviour around $(2{:}2{:}1)$}

From the dynamical viewpoint, the $(2{:}2{:}1)$ heavy top plays a special rôle:
although it is not algebraically integrable, it admits a pure--precession family
arising from an exact curvature balance
\[
K_{\mathrm{geo}}(\Omega_0) + K_{\mathrm{ext}}(\Gamma_0) = 0,
\qquad \|\Gamma_0\|=1,
\]
as shown in \cite{Mityushov221} and reinterpreted in Part~I. This suggests that
small perturbations of the $(2{:}2{:}1)$ inertia ratio should exhibit
near--integrable behaviour organized around this curvature--balanced family.

We briefly outline the qualitative picture, leaving a detailed dynamical
analysis for future work.

\begin{proposition}[Qualitative near--integrable picture near $(2{:}2{:}1)$]
Consider the heavy top with inertia tensor
\[
I(\varepsilon) = \mathrm{diag}\,(2,2,1+\varepsilon),\qquad |\varepsilon|\ll 1,
\]
and fixed centre of mass $\mu=(1,0,0)$. Then:
\begin{enumerate}
\item At $\varepsilon=0$, the system admits a pure--precession family generated by
curvature--balanced initial conditions.

\item For small $\varepsilon\neq 0$, the curvature balance condition is violated
at order $O(\varepsilon)$, and the pure--precession family deforms into a
near--invariant set exhibiting slow drift of the precession axis.

\item The magnitude of this drift is controlled by the curvature deviation
$\Delta(I(\varepsilon))$, which is $O(|\varepsilon|)$, so that the departure
from pure precession evolves on a slow time scale proportional to $1/|\varepsilon|$.
\end{enumerate}
\end{proposition}

\begin{remark}
This proposition is to be understood in a qualitative geometric sense: the
pure--precession family persists as a slow manifold or a family of weakly
invariant tori for small $\varepsilon$, with slow transverse dynamics governed
by the curvature deviation. A detailed analytic treatment would require
normal form and averaging analysis for the Euler--Poisson equations near the
$(2{:}2{:}1)$ curvature--balanced regime, which lies beyond the scope of the
present paper.
\end{remark}

\section{Integrability Map in the $(I_2/I_1,I_3/I_1)$--Plane}

We conclude by visualizing the curvature classification and the curvature deviation
in the $(I_2/I_1,I_3/I_1)$--plane. Following Part~I, we project the space of inertia
ratios onto the plane with coordinates
\[
x = \frac{I_2}{I_1},\qquad y = \frac{I_3}{I_1},
\]
and depict the loci corresponding to the classical integrable cases and to the
mixed ratio $(2{:}2{:}1)$.

\begin{figure}[h]
\centering
\begin{tikzpicture}[scale=1.1]
  \draw[-{Latex}] (-0.2,0) -- (5.2,0) node[below] {$I_2/I_1$};
  \draw[-{Latex}] (0,-0.2) -- (0,3.2) node[left] {$I_3/I_1$};

  \fill (1,1) circle (1.5pt) node[above right] {Spherical};

  \draw[dashed] (1,0.2) -- (1,3);
  \node[above] at (1,3) {Lagrange};

  \draw[dotted] (0.3,0.5) -- (4.8,0.5);
  \node[right] at (4.8,0.5) {Kovalevskaya};

  \draw[dotted] (0.3,0.25) -- (4.8,0.25);
  \node[right] at (4.8,0.25) {Goryachev--Chaplygin};

  \fill (1,0.5) circle (1.5pt) node[above left] {$(2{:}2{:}1)$};

  \node at (3,2) {generic anisotropic};

\end{tikzpicture}
\caption{Curvature integrability map in the $(I_2/I_1,I_3/I_1)$--plane. The spherical point,
Lagrange line, Kovalevskaya and Goryachev--Chaplygin lines, and the mixed $(2{:}2{:}1)$ point
are indicated. The curvature deviation $\Delta(I)$ measures the distance to the integrable
loci and is small in a neighbourhood of these sets, particularly near $(2{:}2{:}1)$.}
\label{fig:integrability-map}
\end{figure}
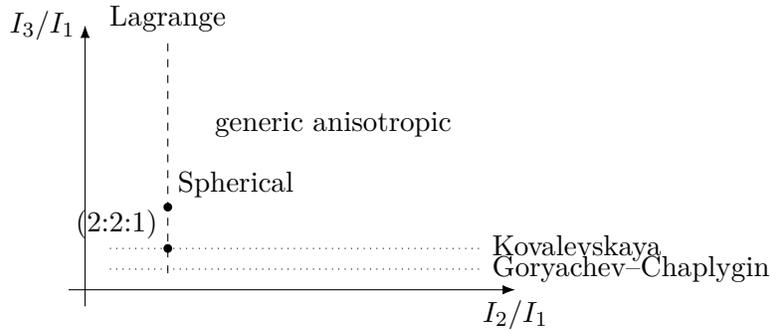

The spherical point $(1,1)$, the Lagrange line $x=1$, and the Kovalevskaya and
Goryachev--Chaplygin lines $y=\tfrac12$ and $y=\tfrac14$ are shown in
Figure~\ref{fig:integrability-map}, together with the mixed $(2{:}2{:}1)$ point at
$(x,y)=(1,\tfrac12)$. The classical integrable cases lie on the union of the spherical
point, the Lagrange line, and the two horizontal lines; the $(2{:}2{:}1)$ point lies at
their intersection, but corresponds to a different curvature signature (balanced--mixed)
than the Kovalevskaya and Goryachev--Chaplygin ratios on the same horizontal lines.

The curvature deviation $\Delta(I)$ can be viewed as defining level sets in this plane:
small values of $\Delta(I)$ correspond to inertia ratios close to one of the integrable
loci, while larger values indicate strongly nonintegrable regimes. In particular, the
neighbourhood of $(2{:}2{:}1)$ provides a natural testing ground for near--integrable
behaviour, governed by the interplay between curvature degeneracy and curvature balance.

\section{Discussion and Outlook}

In this paper we have shown that the classical integrable rigid--body regimes can be
described and unified in terms of curvature signatures of left--invariant metrics on
$SU(2)$. The inertial curvature field $K_{\mathrm{geo}}$ introduced in Part~I provides
a finite--dimensional object whose algebraic degeneracies capture the special inertia
ratios underlying integrability. From this perspective, integrability is not a mysterious
coincidence but a manifestation of simple degeneracies in the curvature structure of the
underlying metric Lie group.

The curvature classification developed here suggests several directions for further work.

First, a more systematic analysis of near--integrable regimes around the $(2{:}2{:}1)$
curvature--balanced case is needed. The curvature deviation $\Delta(I)$ provides a
natural small parameter controlling the departure from integrable signatures, and
normal form or averaging methods for the Euler--Poisson equations near the
pure--precession family of \cite{Mityushov221} should yield precise statements about
the persistence and breakdown of quasi--periodic motions.

Second, the curvature Atlas can be extended to other Lie groups and other mechanical
systems. Left--invariant metrics on higher--dimensional compact Lie groups, and on
the Euclidean group $SE(3)$, possess richer curvature structures, and it would be
natural to seek curvature classifications of integrable and near--integrable regimes
in those settings.

Third, the curvature viewpoint is well suited to control problems. In Part~I a curvature--based
control framework (GCCT) was sketched, in which control inputs are decomposed according to
curvature--induced splittings. The curvature classification of integrable regimes developed
here provides a geometric context for designing control laws that exploit curvature degeneracies
or curvature balance rather than symmetries alone.

Finally, the combination of curvature--based analysis with modern numerical and symbolic tools
opens the possibility of discovering new organized regimes beyond the classical integrable list,
analogous to the $(2{:}2{:}1)$ case. The curvature Atlas thus serves not only as an explanatory
framework but also as a guide for systematic exploration of the space of rigid--body dynamics.


\begin{thebibliography}{99}

\bibitem{Arnold}
V.~I.~Arnold,
\emph{Mathematical Methods of Classical Mechanics},
Springer, 1978.

\bibitem{BolsinovFomenko}
A.~V.~Bolsinov, A.~T.~Fomenko,
\emph{Integrable Hamiltonian Systems: Geometry, Topology, Classification},
Chapman \& Hall/CRC, 2004.

\bibitem{Milnor}
J.~Milnor,
Curvatures of left invariant metrics on Lie groups,
\emph{Advances in Mathematics} \textbf{21} (1976), 293--329.

\bibitem{Kovalevskaya}
S.~V.~Kovalevskaya,
On the motion of a rigid body about a fixed point,
\emph{Math. Ann.} \textbf{16} (1880), 495--517.

\bibitem{Chaplygin}
S.~A.~Chaplygin,
On the motion of a heavy rigid body with a fixed point in a liquid,
\emph{Mat. Sbornik} \textbf{24} (1903), 1--42.

\bibitem{Mityushov221}
E.~A.~Mityushov,
\emph{The (2,2,1) heavy top: a pure-precession regime},
arXiv:2512.05527, 2025.

\bibitem{MityushovAtlasI}
E.~A.~Mityushov,
\emph{Curvature--Driven Dynamics on $S^3$: A Geometric Atlas},
arXiv:2512.14164, 2025.

\end{thebibliography}
\end{document}